\documentclass{imsart}

\usepackage{amsfonts}
\usepackage[fleqn]{nccmath}
\usepackage{graphicx}
\usepackage{tikz-cd}
\usepackage{mathabx}
\usepackage{graphicx}

\RequirePackage[OT1]{fontenc}
\RequirePackage{amsthm,amsmath}
\RequirePackage[numbers]{natbib}
\RequirePackage[colorlinks,citecolor=blue,urlcolor=blue]{hyperref}

\newcommand {\mm}[1] {\ifmmode{#1}\else{\mbox{\(#1\)}}\fi}
\newcommand{\Rspace}{\mm{{\mathbb R}}}

\newcommand{\m}{\mathcal{M}}
\newcommand{\bm}{\partial\mathcal{M}}
\newcommand{\vol}{{\rm vol}}

 \newcommand{\para}[1]        {\noindent{\textbf{\textsf{#1}}}}
\newcommand{\lfs} {\mm{\rm lfs}}
\newcommand{\reach} {\mm{\rm reach}}
\newcommand{\myst} {\mm{\rm st}}
\newcommand{\Hgroup}        {\mm{\sf H}}

\startlocaldefs
\numberwithin{equation}{section}
\theoremstyle{plain}
\newtheorem{theorem}{Theorem}[section]
\newtheorem{lemma}[theorem]{Lemma}

\theoremstyle{remark}
\newtheorem{remark}[theorem]{Remark}
\newtheorem{example}[theorem]{Example}
\endlocaldefs

\begin{document}

\begin{frontmatter}
\title{Topological Inference of Manifolds with Boundary}
\runtitle{Topological Inference of Manifolds with Boundary}

\begin{aug}
\author{\fnms{Yuan} \snm{Wang}\thanksref{t2}\ead[label=e1]{yuanwang@math.northwestern.edu}}
\and
\author{\fnms{Bei} \snm{Wang}\thanksref{t3}\ead[label=e2]{beiwang@sci.utah.edu}}

\address{Northwestern University, Evanston, IL, USA\\
University of Utah, Salt Lake City, UT, USA\\
\printead{e1,e2}}



\thankstext{t2}{NSF research grant DMS-\#1300750 and the Simons Foundation Award \# 256202}
\thankstext{t3}{NSF DBI-1661375 and NSF IIS-1513616}
\runauthor{Y. Wang and B. Wang}

\affiliation{Northwestern University, Evanston, IL, USA and University of Utah, Salt Lake City, UT, USA}

\end{aug}

\begin{abstract}
Given a set of data points sampled from some underlying space, 
there are two important challenges in geometric and topological data analysis when dealing with sampled data: 
reconstruction -- how to assemble discrete samples into global structures, 
and inference -- how to extract geometric and topological information from data that are high-dimensional, incomplete and noisy. 
Niyogi et al.~(2008) have shown that by constructing an offset of the samples using a suitable offset parameter 
could provide reconstructions that preserve homotopy types therefore homology for densely sampled smooth submanifolds of Euclidean 
space without boundary. 
Chazal et al.~(2009) and Attali et al.~(2013) have introduced a parameterized set of sampling conditions that extend 
the results of Niyogi et al.~to a large class of compact subsets of Euclidean space. 
Our work tackles data problems that fill a gap between the work of Niyogi et al.~and Chazal et al. 
In particular, we give a probabilistic notion of sampling conditions for manifolds with boundary that could not be handled by existing theories. 
We also give stronger results that relate topological equivalence between the offset and the manifold as a deformation retract.  
\end{abstract}


\tableofcontents
\end{frontmatter}

\section{Introduction}
\label{sec:introduction}

In \emph{manifold learning}, a topic of high interest is to understand the structure of low-dimensional objects embedded in high-dimensional space. 
Such objects are typically assumed to be (sub)manifolds of  Euclidean spaces. 
In recent years, it is becoming clear that the offsets of sampled data points on a manifold can reflect the geometric and topological structure of the manifold itself (e.g.~\cite{AmentaBern1999, ChengDeyRamos2005, NiyogiSmaleWeinberger2008}). 
In particular, given sampled points drawn from a probability distribution that has support on or near a submanifold without boundary, Niyogi et al.~\cite{NiyogiSmaleWeinberger2008} have shown that one can learn the homology of the submanifold with high confidence. 
More specifically, for a compact manifold $\m$ embedded in Euclidean space $\Rspace^N$ and a set of randomly sampled  data points $\overline{x}=\{x_1,...,x_n\}$ on $\m$, let $U=\bigcup_{x\in\overline{x}}B_{\epsilon}(x)$ be the offset of the data set $\overline{x}$, where $\epsilon$ is chosen to be small relative to the minimum local feature size of $\m$. Then for any $p\in (0,1)$, there is a number $m$ such that for all $n>m$, with probability $p$, $\m$ is a deformation retract of $U$. Therefore the homology of $U$
equals the homology of $\m$ (see \cite[Theorem 3.1]{NiyogiSmaleWeinberger2008} for details). 

Data that arise from smooth compact manifolds have been well-studied. 
However, the study of more complex spaces that are not necessarily manifolds via data samples seems much more difficult. 
When samples arise not from manifolds but from mixtures of manifolds with possible singularities, we are dealing with the notion of \emph{stratification learning}. 
Roughly speaking, a \emph{stratified space} is a space that can be decomposed in to manifold pieces (referred to as \emph{strata}) that are glued together in a nice way. 
The study of stratified spaces is a classic topic in pure mathematics~\cite{GoreskyMacPherson1988, Weinberger1994}. 
Statistical approaches rely on inferences of mixture models, local dimension estimation and subspace clustering~\cite{HaroRandallSapiro2005, LermanZhang2010, VidalMaSastry2005}. 
In geometric and topological data analysis, progress has been made recently in the study of stratified spaces in the discrete and  noisy settings~\cite{BendichCohen-SteinerEdelsbrunner2007, BendichWangMukherjee2012, SkrabaWang2014, BelkinQueWang2012, Nanda2017, BrownWang2018}, 
which draw inspirations from computational topology~\cite{EdelsbrunnerHarer2010}, intersection homology~\cite{Bendich2008, BendichHarer2011, GoreskyMacPherson1982}, graph theory and sheaf theory. 

Among stratified spaces, manifolds with boundary is one of the simplest forms. 
A manifold with boundary is a stratified space: one stratum is its boundary, and the other stratum is the complement. 

In this paper, we study the topology of offsets of data points on compact differentiable manifolds with boundary. 
We give a probabilistic notion of sampling conditions for manifolds with boundary that could not be handled by existing theories. 
In particular, we show that, with some care, a similar statement as \cite[Theorem 3.1]{NiyogiSmaleWeinberger2008} holds for manifolds with boundary. 
We also demonstrate via simple examples how our sampling lower bounds could be derived in practice. 

The main result of this paper, Theorem \ref{main}, is proved by following the framework of~\cite{NiyogiSmaleWeinberger2008}. 
First, we prove that the offset of data points deformation retracts to the manifold if the sample is sufficiently dense with respect to the local feature size of the manifold and the radius of the offset (Theorem \ref{deformretrth}).
Second, we show that such density is achieved with high  confidence when the data points are sufficiently abundant (Theorem \ref{probabestim}). However, our proof in detail is very different from that of \cite{NiyogiSmaleWeinberger2008}: 
particular efforts have been made to overcome the complexity caused by the boundary. 
When a data point is near or on the boundary, the local geometry around it is more complicated, and the original ways~\cite{NiyogiSmaleWeinberger2008} of performing deformation retract and estimating distances and volumes no longer work. We found that, through our arguments, such issues can be resolved by imposing concise and reasonable requirements regarding the minimum local feature size of \emph{both} the manifold and its boundary and the (local) projection of the manifold to its tangent spaces (see Condition \eqref{condition1} and \eqref{condition2}).

It is worth noting that Chazal et al.~\cite{ChazalCohen-SteinerLieutier2009} and Attali et al.~\cite{AttaliLieutierSalinas2013} have extended the results of Niyogi et al. to a large class of compact subsets of Euclidean space. 
Specifically, \cite{ChazalCohen-SteinerLieutier2009} shows that the offset $K^{\beta}$ of a compact set $K$ is homotopy equivalent to the offset $K'^{\alpha}$ of  another compact set $K'$, for sufficiently small $\beta>0$, if $\alpha$ satisfies some inequality involving the $\mu$-reach of $K'$ and the Hausdorff distance between $K$ and $K'$; \cite{AttaliLieutierSalinas2013} shows that the offset $X^{\eta}$ of a compact set $X$ is homotopy equivalent to the \v{C}ech (respectively the Rips) complex of certain radius of a point cloud $P$ in $X$ if some inequality involving $\eta$, the radius and the $\mu$-reach of $X$ holds. Compared to the result of~\cite{ChazalCohen-SteinerLieutier2009} and \cite{AttaliLieutierSalinas2013}, the one in this paper only deals with compact manifolds with boundary. However, our result does have its own advantages. First of all, the topological equivalence between the offset of data points and the manifold we get here is a deformation retract, which is stronger than homotopy equivalence derived in \cite{ChazalCohen-SteinerLieutier2009} and \cite{AttaliLieutierSalinas2013}. Next, we get a probability estimate for the topological equivalence which is not provided in \cite{ChazalCohen-SteinerLieutier2009} and \cite{AttaliLieutierSalinas2013}. Last but not least, it appears that there are  elementary cases of data on manifolds with boundary where the result in \cite{ChazalCohen-SteinerLieutier2009} does not apply, as the parameters associated with the data are completely outside the scope of \cite[Theorem 4.6]{ChazalCohen-SteinerLieutier2009} and \cite[Theorem 13 and 14]{AttaliLieutierSalinas2013}. A scenario is discussed in Section \ref{sec:results} (Example \ref{exbeyondccsl}), where our Theorems \ref{main} and \ref{deformretrth} 
become applicable and work well. 

\para{Result at a glance.}
In short, this paper improves our understanding of topological inference for manifolds with boundary and 
therefore enriches the toolbox for topological data analysis. 
Given a sample of $n$ points from a differentiable manifold $\m$ with boundary 
in a high-dimensional Euclidean space, for a sufficiently large $n$, 
the $\epsilon$-offset of the sample points is shown to have the same homotopy type as and deformation retracts to $\m$.
The homotopy equivalence result has been proved by Niyogi et al.~\cite{NiyogiSmaleWeinberger2008} for manifolds without boundary. 
Chazal et al.~\cite{ChazalCohen-SteinerLieutier2009} have extended the result of Niyogi et al.~to manifolds with boundary, 
however with weaker conclusions than the original paper. 
The current paper, instead, reaches the same conclusion as Niyogi et al.~\cite{NiyogiSmaleWeinberger2008} for manifold with boundary, 
but under two mild conditions for the boundary (i.e.,~regarding the minimum feature size of the boundary and the uniform smoothness of the tangent bundle).
A specific example shows that the new method is more powerful than that of Chazal et al.~\cite{ChazalCohen-SteinerLieutier2009}. 
Therefore, our results fill a gap between the original work of Niyogi et al.~\cite{NiyogiSmaleWeinberger2008} and 
the broader but weaker result of Chazal et al.~\cite{ChazalCohen-SteinerLieutier2009}.  
Our proof techniques are nontrivial, although they frequently use the results of Niyogi et al.~and make necessary adjustments. 
The theoretical results are complemented by experiments that confirm the theoretical findings.

\section{Notations and preliminaries}
\label{sec:prelim}

\subsection{Basics on manifolds}

In this paper, for two points $p,q\in \Rspace^N$, we use $|pq|$ to denote the line segment connecting $p$ and $q$, $\overrightarrow{pq}$ to denote the vector from $p$ to $q$, and $\|p-q\|$ to denote the Euclidean distance between $p$ and $q$. For a set $K\subset \Rspace^N$, $d(p,K):=\inf\{\|p-q\|: q\in K\}$ denotes the distance from $p \in \Rspace^N$ to $K$. Moreover, for a non-negative real number $\alpha \in \Rspace_{\geq 0}$, we use $K^{\alpha}$ to denote the offset of $K$ with radius $\alpha$, 
$K^{\alpha} = \{ p \in \Rspace^N \mid d(p,K) \leq \alpha\}$. 
$B_{r}(p)$ denotes the open ball with center $p$ and radius $r$. 

Let $\mathcal{M}\subseteq\Rspace^N$ be a compact, \emph{differentiable},  $k$-dimensional manifold possibly with boundary. Let $\bm$ denote the boundary of $\m$. Then $\bm$ is a compact manifold. Let $\m^{\circ}$ denote the interior of $\m$. 

The \emph{local feature size} of $\m$ is the function $\lfs: \m \to \Rspace_{\geq 0}$ defined by the distance from a point $x \in \m$ to its medial axis. 
The infimum of $\lfs$ is the reach of $\m$, $\reach(\m)$.
For every number $0<r<\reach(\m)$, the normal bundle about $\m$ of radius $r$ is embedded in $\Rspace^N$. In the same way we define $\reach(\bm)$, and since $\bm$ is also a compact manifold, $\reach(\bm)$ is well-defined.


We use $\varphi_{p, \m}$ to denote the natural projection from $\m$ to $T_p(\m)$, the tangent space to $\m$ at the point $p$. That is, 
 $\varphi_{p, \m}: \m \to T_p(\m)$.  
Conversely for any point $q\in T_p(\m)$ we use $\varphi_{p, \m}^{-1}(q)$ to denote the set of points in $\m$ which maps to $q$ via $\varphi_{p, \m}$.

We take $\delta=\delta(\mathcal{M}) \in \Rspace_{\geq 0}$ to be any non-negative real number such that for any $p\in\m$, the following Condition~\ref{condition1} and Condition~\ref{condition2} are satisfied: 
\begin{ceqn}
\begin{align}
& \delta  <\min\{\reach(\m), \reach(\bm)\}, \label{condition1}\\
& \varphi_{p, \m}|_{B_{\delta}(p)\cap\mathcal{M}}  \textnormal{ is a diffeomorphism onto its image. } \label{condition2}
\end{align}
\end{ceqn}

Finally, suppose $\overline{x}=\{x_1,...,x_n\}$ is a set of sampled data points from a compact, differentiable manifold $\m$ with boundary. 
$\overline{x}$ is \emph{$\epsilon$-dense} if for any $p \in \m$, there is a point $x \in \overline{x}$ such that $x \in B_{\epsilon}(p)$. 
Let $U=\bigcup_{x\in\overline{x}}B_{\epsilon}(x)$ denote the offset of $\overline{x}$. 
We also define the canonical map $\pi:U\to \m$ by 
$$\pi(x):={\rm arg}\min_{p\in\m}\|x-p\|.$$
By Condition (1) and the definition of $\reach$, $\pi$ is well-defined.

\subsection{Volume of a hyperspherical cap} 
\label{sechypersphecap}

Let $S$ be a $k$-dimensional hypersphere of radius $r$. Let $H$ be a hyperplane that divides $S$ into two parts. We take the smaller part as a hyperspherical cap. Let $a$ be the radius of the base of the cap, and $\phi:=\arcsin(\dfrac{a}{r})$. Then by \cite{Li2011}, the volume of the hypershperical cap is 
$$V=\frac{2\pi^{\frac{k-1}{2}}r^k}{\Gamma(\frac{k+1}{2})}=\frac{1}{2}V_k(r)I_{\sin^2\phi}(\frac{k+1}{2},\frac{1}{2}),$$
where $\Gamma(x)$ is the Gamma function, $V_k(r)$ is the volume of the $k$-dimensional sphere with radius $r$ and $I_x(a,b)$ is the regularized incomplete beta function.

\section{The main results}
\label{sec:results}
The main theorem of this paper (Theorem~\ref{main}) is centered around a probabilistic notion of sampling conditions for manifolds with boundary, which relates the topological equivalence between the offset (of samples) and the manifold as a deformation retract. 
To the best of our knowledge, such a result has not been addressed by existing theories. 

\begin{theorem} \label{main}
Let $\mathcal{M}\subseteq\mathbb{R}^N$ be a compact differentiable $k$-dimensional manifold possibly with boundary. Let $\overline{x}=x_1, x_2, ..., x_n$ be drawn by sampling $\m$ in i.i.d fasion according to the uniform probability measure on $\m$. Let $\epsilon\in (0,\dfrac{1}{2}\delta(\m))$ and $U=\bigcup_{x\in\overline{x}}B_{\epsilon}(x)$. Then for all 
$$n> \beta(\epsilon)(\ln \beta(\frac{\epsilon}{2})+\ln(\frac{1}{\gamma})),$$
$U$ deformation retracts to $\m$ with probability $>1-\gamma$.  
In particular, with such confidence, the homology of $U$ is the same as that of $\m$. 
Here, $\displaystyle\beta(x):=\frac{\vol(\m)}{\frac{\cos^k\theta}{2^{k+1}}I_{1-\frac{x^2\cos^2\theta}{16\delta^2}}(\frac{k+1}{2}, \frac{1}{2})\vol(B_{x}^k)}$ and $\theta=\arcsin(\frac{x}{4\delta})$.
\end{theorem}

Theorem~\ref{main} is implied by combining Theorems \ref{deformretrth} and \ref{probabestim} below. 

\begin{theorem} 
\label{deformretrth}
Let $\bar{x}$ be any finite collection of points $x_1, ..., x_n\in \mathbb{R}^{N}$ such that it is $\dfrac{\epsilon}{2}$-dense in $\m$. Then for any $\epsilon<\dfrac{\delta}{2}$, we have that $U$ deformation retracts to $\m$. 
\end{theorem}

\begin{theorem} 
\label{probabestim}
Let $\overline{x}$ be drawn by sampling $\m$ in i.i.d fasion according to the uniform probability measure on $\m$. Then with probability $1-\gamma$, we have that $\overline{x}$ is $\dfrac{\epsilon}{2}$-dense ($\epsilon<\frac{\delta}{2})$ in $\m$ provided
$$|\overline{x}|\ge \beta(\epsilon)(\ln \beta(\frac{\epsilon}{2})+\ln(\frac{1}{\gamma})).$$
\end{theorem}

As is mentioned in Section~\ref{sec:introduction}, compared to \cite[Theorem 4.6]{ChazalCohen-SteinerLieutier2009} and \cite[Theorem 13]{AttaliLieutierSalinas2013}, Theorem \ref{deformretrth} only deals with manifolds with boundary, but it establishes a criterion for deformation retract, which is stronger than homotopy equivalence as in \cite{ChazalCohen-SteinerLieutier2009}. The following is an example where neither \cite[Theorem 4.6]{ChazalCohen-SteinerLieutier2009} or \cite[Theorem 13]{AttaliLieutierSalinas2013} applies but Theorem \ref{deformretrth} does. 

\begin{example} 
\label{exbeyondccsl}
Let $C$ be a semi-unit-circle. Then $C$ is a manifold with boundary and its boundary consists of the two end points which are also the end points of a diameter. Let the data be the points $A_1,A_2,...,A_8$ which divide the semi-circle evenly into 7 arcs, with $A_1$ and $A_8$ being the end points. Denote $A:=\{A_1,...,A_8\}$.

We first treat this case with the result in \cite{ChazalCohen-SteinerLieutier2009}. We adopt the notations there. \cite[Theorem 4.6]{ChazalCohen-SteinerLieutier2009} is the main reconstruction theorem, and it requires the inequality 
\begin{ceqn}
\begin{align}
\label{resultccsl}
\frac{4d_H(K,K')}{\mu^2}\le \alpha<r_{\mu}(K')-3d_H(K,K')
\end{align}
\end{ceqn}
in order for the offset $K^{\alpha}$ to be homotopy equivalent to $K'$. Here we let $K'=C$ and $K=A$. It is easy to see that $d_H(K,K')=2\sin\dfrac{\pi}{28}\approx 0.223928$, and
\begin{ceqn}
\begin{align*}
 r_{\mu}(K')=r_{\mu}(C)=\begin{cases} 
      0 & \mu> 1 \\
      1 & 0<\mu\le 1. 
      \end{cases}
\end{align*}
\end{ceqn}
If $\mu>1$, then as $r_{\mu}(K')=0$, the right half of \eqref{resultccsl} does not make sense. If $0<\mu\le 1$, then 
$$\dfrac{4d_H(K,K')}{\mu^2}\ge 4d_H(K,K')=8\sin\dfrac{\pi}{28}>1-6\sin\dfrac{\pi}{28}=r_{\mu}(K')-3d_H(K,K').$$
This is a contradiction to \eqref{resultccsl}. Therefore in this case, \cite[Theorem 4.6]{ChazalCohen-SteinerLieutier2009} does not apply.

Next we try \cite[Theorem 13 and 14]{AttaliLieutierSalinas2013}. \cite[Theorem 13]{AttaliLieutierSalinas2013} requires that 
$$d_H(A,C)\le \epsilon <\lambda^{\rm cech}(\mu)r_{\mu}(C),$$
where 
$$\lambda^{\rm cech}(\mu)=\dfrac{-3\mu+3\mu^2-3+\sqrt{-8\mu^2+4\mu^3+18\mu+2\mu^4+9+\mu^6-4\mu^5}}{-7\mu^2+22\mu+\mu^4-4\mu^3+1}$$
We have $d_H(K,K')=0.223928$ as well as the value of $r_{\mu}(C)$ as is deduced above. On the other hand, from \cite[Fig. 9]{AttaliLieutierSalinas2013} we see that $\lambda^{\rm cech}(\mu)$ is increasing on (0,1], hence $\lambda^{\rm cech}(\mu)\le \lambda^{\rm cech}(1)=\dfrac{-3+\sqrt{22}}{13}\approx 0.130032$. So it is easy to see that $d_H(A,C)>\lambda^{\rm cech}(\mu)r_{\mu}(C)$ for all $\mu\in \mathbb{R}_+$. Therefore \cite[Theorem 13]{AttaliLieutierSalinas2013} does not apply. \cite[Theorem 14]{AttaliLieutierSalinas2013} requires that
$$d_H(A,C)\le \epsilon <\lambda_{n}^{\rm rips}(\mu)r_{\mu}(C).$$
But by \cite[Fig. 9]{AttaliLieutierSalinas2013}, $\lambda_{n}^{\rm rips}(\mu)$ is always smaller than $\lambda^{\rm cech}(\mu)$, so \cite[Theorem 14]{AttaliLieutierSalinas2013} does not apply either. 

Finally we try to apply Theorem \ref{deformretrth} to this case. It is easy to see that $\reach(C)=\reach(\partial C)=1$, and it also satisfies Condition \eqref{condition2} to let $\delta(C)=1$. Now let $\epsilon=0.48<\dfrac{1}{2}=\dfrac{\delta(C)}{2}$. Since $d_H(A,C)=2\sin\dfrac{\pi}{28}\approx 0.223928<0.24=\dfrac{\epsilon}{2}$ as is calculated above, $A$ is $\dfrac{\epsilon}{2}$-dense. Therefore Theorem \ref{deformretrth} applies. 
\end{example}

\section{Proofs of the main results}
\label{sec:proofs}
To prove our main results, we begin with a series of lemmas. Recall the canonical map $\pi:U\to \m$ is defined by 
$\pi(x):={\rm arg}\min_{p\in\m}\|x-p\|$. Lemmas \ref{allintubneighbor}-\ref{star} contribute to the proof of Theorem \ref{deformretrth}. 
Roughly speaking, they show that $\pi^{-1}(p)$ is star-shaped for every $p\in\m$, hence the deformation retract is well-defined. 
Lemmas \ref{localvolest}-\ref{volest} contribute to the proof of Theorem \ref{probabestim} by giving a lower bound to the volume of $\m\cap B_{\epsilon}(p)$ for every $p\in\m$ in terms of $\epsilon$. Then by directly applying results in \cite{NiyogiSmaleWeinberger2008}, we get the estimation of the number of data points as in Theorem \ref{probabestim}.

\begin{lemma} \label{allintubneighbor}
Choose $\delta=\delta(\m)>0$ as in Section~\ref{sec:prelim}. Then for any $\lambda\in(0,\delta)$, any point $q$ such that $\|q-\pi(q)\|<\lambda$ and any $q'\in |q\pi(q)|$, we have $\pi(q')=\pi(q)$. 
\end{lemma}

\begin{proof}
Suppose that this is not the case. Then 
$$\|q-\pi(q)\|=\|q-q'\|+\|q'-\pi(q)\|>\|q-q'\|+\|q'-\pi(q')\|\ge \|q-\pi(q')\|.$$
This is a contradiction to Condition \ref{condition1}. 
\end{proof}

Suppose a certain point $q \in U$ deformation retracts to a point $\pi(q) \in \m$. 
Let $|q \pi(q)|$ be the path of the deformation retract for $q$. 
Lemma \ref{allintubneighbor} tells us that along the path $|q \pi(q)|$, all the points deformation retracts to $\pi(q)$.

From now on we set $\epsilon<\dfrac{\delta}{2}$. For a point $p\in \m$ and a point $p'\in \bm$ we define $\myst_{\m}(p)$ and $\myst_{\bm}(p')$ as
\begin{ceqn}
\begin{align*}
& \myst_{\m}(p):=\bigcup_{x\in\bar{x};x\in B_{\epsilon}(p)}(B_{\epsilon}(x)\cap T_p(\m)^{\perp}). \\
& \myst_{\bm}(p'):=\bigcup_{x\in\bar{x};x\in B_{\epsilon}(p')}(B_{\epsilon}(x)\cap T_{p'}(\bm)^{\perp}).
\end{align*} 
\end{ceqn}
For convenience, we also define $\myst_{\bm}(p'):=\emptyset$ if $p' \in \m^{\circ}$. 
We present the next lemma, whose proof is exactly the same as those corresponding ones in \cite{NiyogiSmaleWeinberger2008}, although the statements are somewhat different.
\begin{lemma} [\cite{NiyogiSmaleWeinberger2008}, Proposition 4.1] \label{ststarshaped}
$\myst_{\m}(p)$ (resp.~$\myst_{\bm}(p)$) is star-shaped for any $p\in\m$ (resp.~$p\in\bm$).
\end{lemma}

\begin{lemma} 
\label{pi-1contained}
$\pi^{-1}(p)\subseteq \myst_{\m}(p)$ (resp. $\pi^{-1}(p)\subseteq \myst_{\bm}(p)$) for any $p\in\m$ (resp. $p\in\bm$).
\end{lemma}

\begin{proof}
If $p \in  \m^{\circ}$, 
this is already proven in \cite[Proposition 4.2]{NiyogiSmaleWeinberger2008}. So for the rest of the proof we assume that $p\in\bm$.

Let $v$ be an arbitrary point in $\pi^{-1}(p)$. By the definition of $\myst_{\bm}$, we only need to consider the case where there is a point $q\in \bar{x}$ such that $q\notin B_{\epsilon}(p)$ and $v\in B_{\epsilon}(q)$. In this case, the distance between $v$ and $p$ is at most $\dfrac{\epsilon^2}{\delta}$, and the proof is exactly the same as that of \cite[Lemma 4.1]{NiyogiSmaleWeinberger2008}. Now by the $\dfrac{\epsilon}{2}$-dense condition, there is a point $x\in\overline{x}$ such that $\|x-p\|\le \dfrac{\epsilon}{2}$. Therefore 
$$\|v-x\|\le \|v-p\|+\|p-x\|\le \dfrac{\epsilon^2}{\delta}+\dfrac{\epsilon}{2}< \dfrac{\epsilon}{2}+\dfrac{\epsilon}{2}=\epsilon.$$
\end{proof}
\begin{figure}[!ht]
\begin{center}
\begin{tabular}{c}
\includegraphics[width=0.45\linewidth]{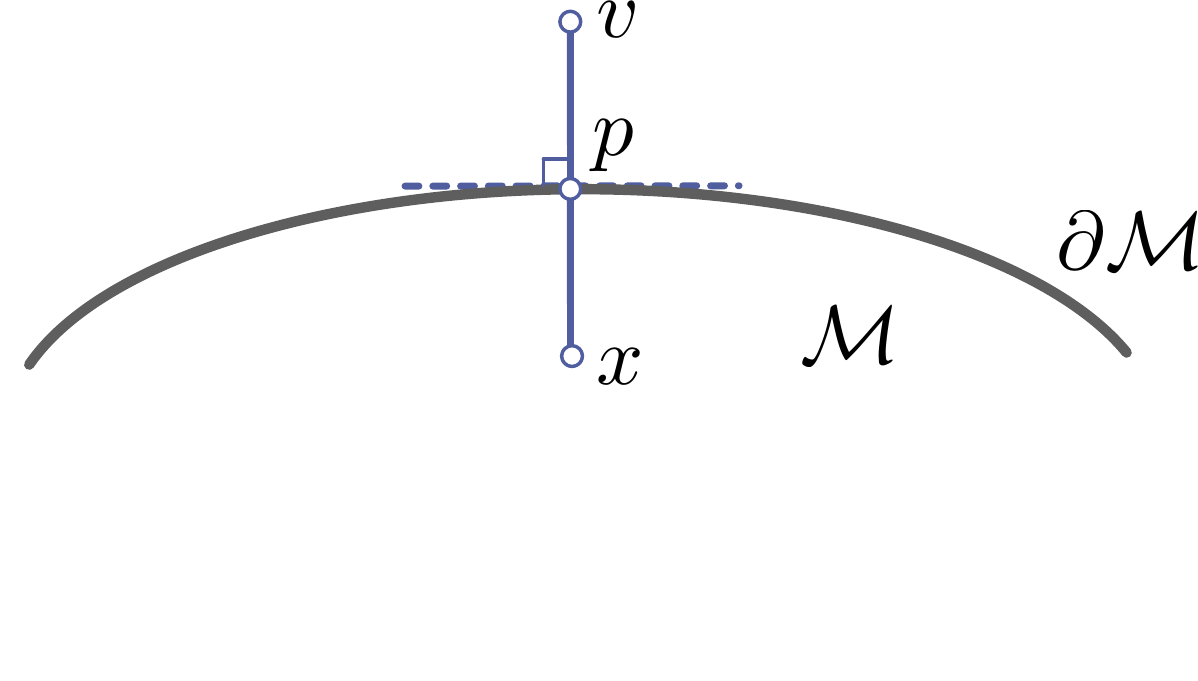} 
\end{tabular}
\vspace{-12mm}
\caption{The worse case for Lemma~\ref{pi-1contained}, where $||v-p|| = ||p-x|| = \frac{\epsilon}{2}$.}
\label{fig:lemma7}
\end{center}
\end{figure}
\begin{remark}
We note that the worst case in Lemma \ref{pi-1contained} can happen when $p$ lies on the boundary of $\m$. This is illustrated in Fig.~\ref{fig:lemma7}, where $M$ is a plane region with boundary, $v$ is also inside the plane, and $|vp|$ and $|xv|$ are both perpendicular to $\bm$.
This is why we require that $\epsilon<\dfrac{\delta}{2}$, which is more restrictive than the requirement $\epsilon<\sqrt{\dfrac{5}{3}}\delta$ in \cite{NiyogiSmaleWeinberger2008}.
\end{remark}
\begin{lemma} \label{star}
For any $p\in\m$, $\pi^{-1}(p)$ is star-shaped with respect to $p$.
\end{lemma}
\begin{proof}
If $p \in \m^{\circ}$, 
we have by Lemma \ref{ststarshaped} and \ref{pi-1contained} that $\pi^{-1}(p)=st_{\m}(p)$ which is star-shaped. So for the rest of the proof we assume that $p\in\bm$. By Lemma \ref{ststarshaped}, $st_{\bm}(p)$ is star-shaped; by Lemma \ref{allintubneighbor}, for any ray $l\subset T_p(\bm)^{\perp}$ starting from $p$, either $\pi^{-1}(p)\cap l=\{p\}$ or $st_{\bm}(p)\cap l\subseteq \pi^{-1}(p)\cap l$; by Lemma \ref{pi-1contained} we know that $\pi^{-1}(p)\subseteq st_{\bm}(p)$. So if $q\in\pi^{-1}(p)$, then $|pq|\subseteq\pi^{-1}(p)$.
\end{proof}

Then Theorem~\ref{deformretrth} generalizes \cite[Proposition 3.1]{NiyogiSmaleWeinberger2008} to compact manifolds with boundaries. 
Its proof is as follows. 

\begin{proof}
We define the deformation retract $F(x,t): U\times [0,1]\to U$ as $F(x,t)=tx+(1-t)\pi(x)$. By Lemma \ref{allintubneighbor} and Lemma  \ref{star}, this deformation retract is well-defined. Moreover since $U$ is contained in $\m^{\delta}$, there is no critical point for distance functions, we get that $U$ deformation retracts to $\m$.
\end{proof}
\begin{lemma} \label{localvolest}
Let $p\in\bm$ and $A=\m\cap B_{\epsilon}(p)$, where $\epsilon\in(0,\delta)$. Then 
$$\vol(A)> \frac{\cos^k\theta}{2}I_{1-\frac{\epsilon^2\cos^2\theta}{4\delta^2}}(\frac{k+1}{2}, \frac{1}{2})\vol(B_{\epsilon}^k(p)),$$
where $I$ is the regularized incomplete beta function, $B_{\epsilon}^k(p)$ is the $k$-dimensional ball in $T_p$ centered at $p$, and $\theta=\arcsin(\dfrac{\epsilon}{2\delta})$.
\end{lemma}
\begin{proof}
Let $p'\in T_p(\m)$ be the point such that $\overrightarrow{pp'}$ is perpendicular to $T_p(\bm)$ and points to the inside of $\m$, and $\|p-p'\|=\delta$. We first want to show that 
\begin{ceqn}
\begin{align} \label{inclintersball}
B_{\epsilon\cos\theta}^k(p)\cap B_{\delta}^k(p')\subseteq \varphi_{p,\m}(A).
\end{align}
\end{ceqn}
By Condition \eqref{condition2}, $\varphi_{p, \m}|_{A}$ is a homeomorphism onto its image. In particular, 
$$\varphi_{p,\m}(\partial A)=\partial (\varphi_{p,\m}(A)).$$ 
It is easy to see that $B_{s}^k(p)\cap B_{\delta}^k(p')\cap \varphi_{p,\m}(A)\ne \emptyset$ for any $s>0$. Let $q\in\m$ be a point on the boundary of $A$. Then $q\in \bm\cup\partial B_{\epsilon}(p)$. We prove \eqref{inclintersball} by proving the claim that no matter whether $q\in \bm$ or $q\in \partial B_{\epsilon}(p)$, $\varphi_{p,\m}(q)$ is outside $B_{\epsilon\cos\theta}^k(p)\cap B_{\delta}^k(p')$. Indeed, if there exists a point $q'\in B_{\epsilon\cos\theta}^k(p)\cap B_{\delta}^k(p')$ such that $q'\not\in \varphi_{p,\m}(A)$, we choose a point $o\in B_{\epsilon\cos\theta}^k(p)\cap B_{\delta}^k(p')\cap \varphi_{p,\m}(A)$. Then the line segment connecting $o$ and $q'$ must intersect with the $\partial(\varphi_{p,\m}(A))$. But on the other hand, by the convexity of $B_{\epsilon\cos\theta}^k(p)\cap B_{\delta}^k(p')$, any intersection point is inside $B_{\epsilon\cos\theta}^k(p)\cap B_{\delta}^k(p')$. This is a contradiction.

To prove the claim, we first suppose that $q\in \bm$. Let $p''$ be the point where $\overrightarrow{pp''}$ is in the same direction with $\overrightarrow{\varphi_{p,\bm}(q)q}$ and $\|p''-p\|=\delta$ (if $\varphi_{p,\bm}(q)=q$, then set $p''=p'$). By Condition \eqref{condition1}, $q$ is outside $B_{\delta}(p'')$. Now $q$, $\varphi_{p,\m}(q)$ and $\varphi_{p,\bm}(q)$ form a right-angled triangle where $|q\varphi_{p,\bm}(q)|$ is the hypotenuse, so $\|\varphi_{p,\m}(q)-\varphi_{p,\bm}(q)\|\le \|q-\varphi_{p,\bm}(q)\|$. Therefore $\varphi_{p,\m}(q)$ is certainly outside $B_{\delta}^k(p')$. This case is illustrated in Fig.~\ref{fig:lemma9a}. 

\begin{figure}[!ht]
\begin{center}
\begin{tabular}{c}
\includegraphics[width=0.5\linewidth]{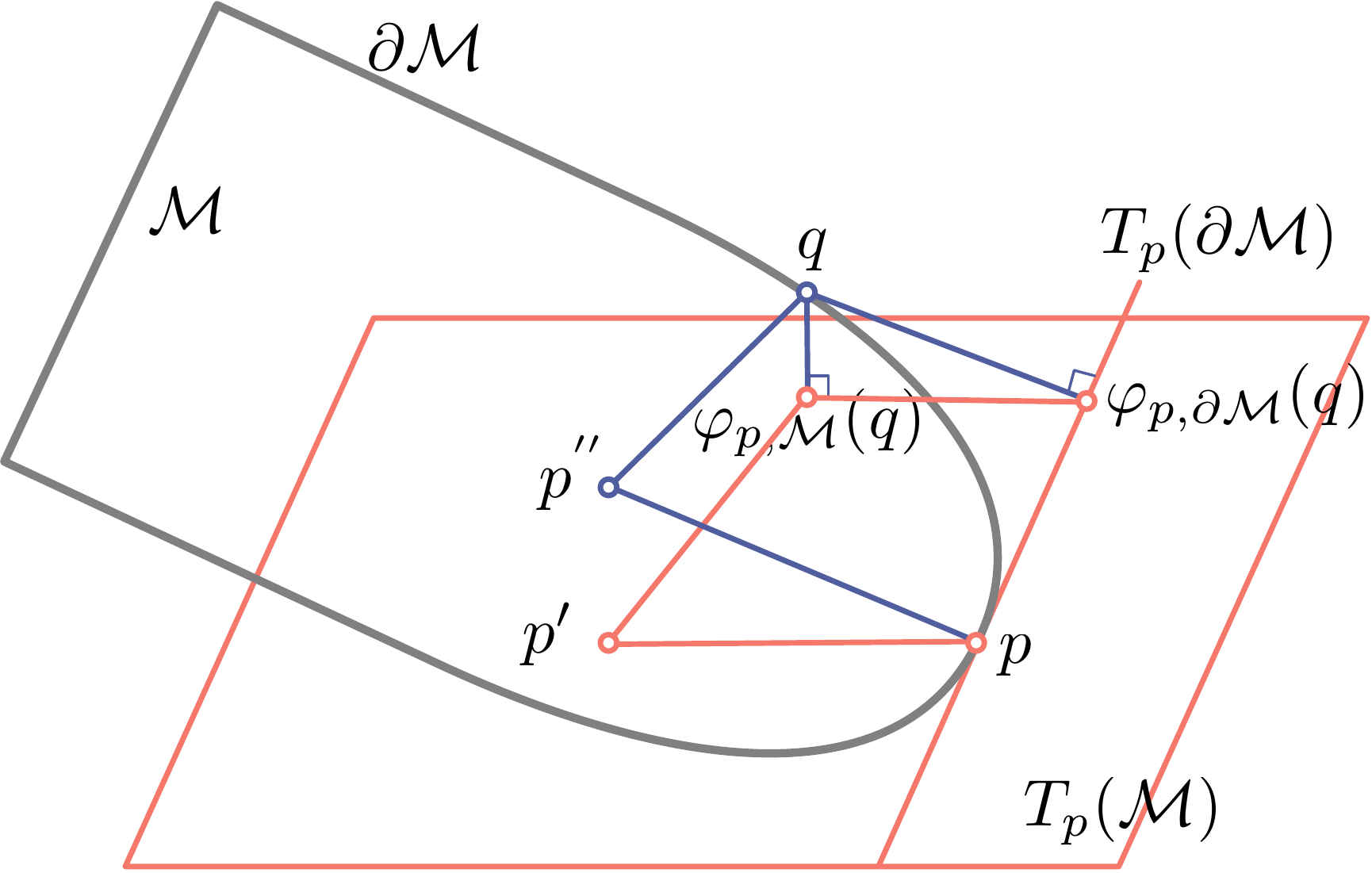} 
\end{tabular}
\vspace{-2mm}
\caption{An illustration for Lemma~\ref{localvolest} for the case when $q \in \partial \m$, where $||p' - p|| = \delta$, $||p'' - p|| = \delta$, $||p^{''} - q|| > \delta$, $||q - \varphi_{p, \partial \m}(q)|| \leq \delta$, therefore $||p' - \varphi_{p, \m}(q)|| \geq ||p^{''} - q|| > \delta$.}
\label{fig:lemma9a}
\end{center}
\end{figure}

Next suppose that $q\in \partial B_{\epsilon}(p)$, which implies that $\|q-p\|=\epsilon$. We have that $\|\varphi_{p,\m}(q)-p\|=\epsilon\cos(\phi)$ where $\phi$ is the (smallest nonnegative) angle between $\overrightarrow{pq}$ and $\overrightarrow{p\varphi_{p,\m}(q)}$. Let $p''$ be the point where $\overrightarrow{pp''}$ is in the same direction with $\overrightarrow{\varphi_{p,\m}(q)q}$ (if $\varphi_{p,\m}(q)=q$, then choose an arbitrary $p''$ such that $pp''\bot T_{p}(\m)$) and $\|p''-p\|=\delta$. Then by Condition \eqref{condition1}, $\|p''-q\|>\delta$. So by the definition of $\theta$ we see that $\phi<\theta$, and $\|\varphi_{p,\m}(q)-p\| = \epsilon\cos (\phi) > \epsilon\cos(\theta)$. Hence $\varphi_{p,\m}(q)$ is outside $B_{\epsilon\cos\theta}^k(p)$. This is illustrated in Fig.~\ref{fig:lemma9b}.

\begin{figure}[!ht]
\begin{center}
\begin{tabular}{c}
\includegraphics[width=0.5\linewidth]{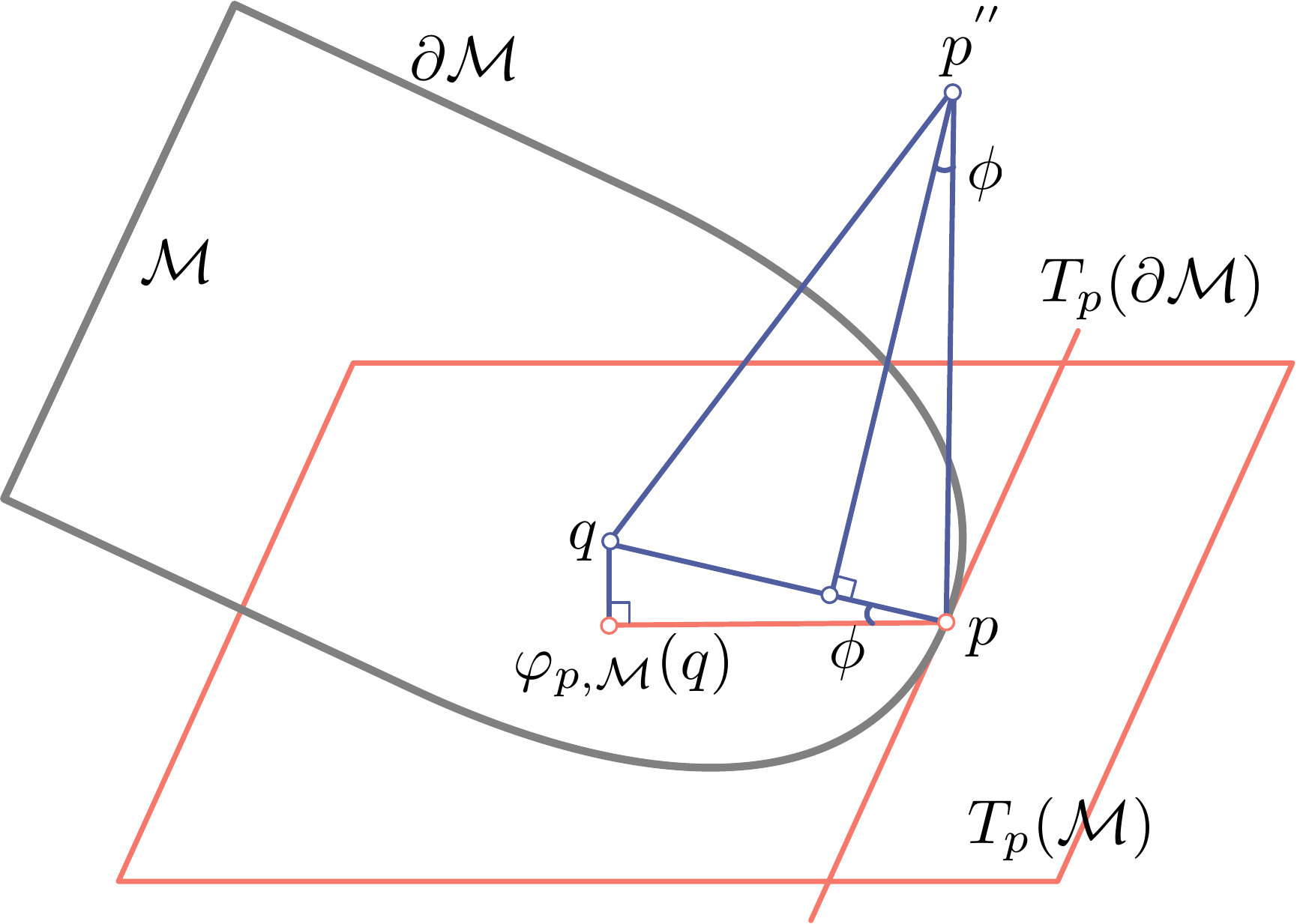} 
\end{tabular}
\vspace{-2mm}
\caption{An illustration for Lemma~\ref{localvolest} for the case when $q \in \partial B_{\epsilon}(p)$, where $||p-q|| = \epsilon$, 
$\phi < \theta$, $|p^{''} - p|| = \delta$ and $||p^{''} - q|| > \delta$.}
\label{fig:lemma9b}
\end{center}
\end{figure}

So we have $\vol(A)\ge \vol(B_{\epsilon\cos\theta}^k(p)\cap B_{\delta}^k(p'))$. The right-hand side consists of two hyperspherical caps. For convenience we choose the lower bound of the right side to be the hyperspherical cap that belongs to $B_{\epsilon\cos\theta}^k(p)$. By \ref{sechypersphecap}, the volume of this hyperspherical cap is $\displaystyle\frac{1}{2}I_{1-\frac{\epsilon^2\cos^2\theta}{4\delta^2}}(\frac{k+1}{2}, \frac{1}{2})\vol(B_{\epsilon\cos\theta}^k(p))$. Moreover we know that $\vol(B_{\epsilon\cos\theta}^k(p))=\cos^k\theta\vol(B_{\epsilon}^k(p))$, so we are done.
\end{proof}
Using the same argument as in the third paragraph of the proof of the last lemma, we actually have 
\begin{lemma} \label{estnointersectionwithbound}
Let $p\in\m$ and $\epsilon>0$ such that $\bm\cap B_{\epsilon}(p)=\emptyset$. Let $A:=\m\cap B_{\epsilon}(p)$. Then 
$${\rm vol}(A)\ge \vol(B_{\epsilon\cos\theta}^k(p))=\cos^k\theta\vol(B_{\epsilon}^k(p)).$$
\end{lemma}
Combining Lemma \ref{localvolest} and \ref{estnointersectionwithbound}, we have
\begin{lemma} \label{volest}
Let $p\in\m$ and $A=\m\cap B_{\epsilon}(p)$, where $\epsilon\in(0,\delta)$. Then 
$${\rm vol}(A)\ge\frac{\cos^k\theta'}{2^{k+1}}I_{1-\frac{\epsilon^2\cos^2\theta'}{16\delta^2}}(\frac{k+1}{2}, \frac{1}{2})\vol(B_{\epsilon}^k(p)),$$
where $\theta':=\arcsin(\dfrac{\epsilon}{4\delta})$.
\end{lemma}
\begin{proof}
If $d(p,\bm)>\dfrac{\epsilon}{2}$, then $\bm\cap B_{\frac{\epsilon}{2}}(p)=\emptyset$. So by Lemma \ref{estnointersectionwithbound}, 
$${\rm vol}(A)\ge \cos^k\theta'\vol(B_{\frac{\epsilon}{2}}^k(p))$$
and we are done. If $d(p,\bm)\le \dfrac{\epsilon}{2}$, let $p'$ be a point on $\bm$ that has minimum distance from $p$. Then $B_{\frac{\epsilon}{2}}(p')\subset B_{\epsilon}(p)$. So 
\begin{align*}
{\rm vol}(A)\ge{\rm vol}(B_{\frac{\epsilon}{2}}(p')\cap \m)& \ge \frac{\cos^k\theta'}{2}I_{1-\frac{\epsilon^2\cos^2\theta'}{16\delta^2}}(\frac{k+1}{2}, \frac{1}{2})\vol(B_{\frac{\epsilon}{2}}^k(p)) \\
& =\frac{\cos^k\theta'}{2^{k+1}}I_{1-\frac{\epsilon^2\cos^2\theta'}{16\delta^2}}(\frac{k+1}{2}, \frac{1}{2})\vol(B_{\epsilon}^k(p)),
\end{align*}
where the last inequality is by Lemma \ref{localvolest}.
\end{proof}
We observe that the right side of the inequality in Lemma \ref{volest} is $\dfrac{\vol(\m)}{\beta(\epsilon)}$, where the function $\beta$ is defined in Theorem~\ref{main} (note that the $\theta'$ here corresponds to the $\theta$ in Theroem~\ref{main}). By \cite[Lemma 5.1 and 5.2]{NiyogiSmaleWeinberger2008}, a satisfactory number of points $|\overline{x}|$ as in Theorem~\ref{probabestim} is of the form $\dfrac{1}{\alpha}(\ln l+\ln \dfrac{1}{\gamma})$, where $\alpha$ is any lower bound of $\dfrac{\vol(A)}{\vol(\m)}$ and $l$ is any upper bound of $\dfrac{\epsilon}{2}$-packing-number. 
So by Lemma \ref{volest}, it suffices to take $\alpha$ and $l$ to be $\dfrac{1}{\beta(\epsilon)}$ and $\beta(\dfrac{\epsilon}{2})$, respectively. Therefore we obtain Theorem~\ref{probabestim}.
Finally, combining Theorems \ref{deformretrth} and \ref{probabestim}, we arrive at Theorem~\ref{main}.

\section{Experiments} 
\label{sec:experiments}
In this section, we work on two typical examples of manifolds with boundary. The first example is a cylindrical surface, referred to as the \emph{cylinder} dataset,  which has radius $1$ and height $1$. More precisely, it can be expressed as
$$\{(x,y,z)\in \mathbb{R}^3: x^2+y^2=1, z\in [0,1]\}.$$

The second example is a torus with a cap chopped off, referred to as the \emph{torus} dataset. In $\mathbb{R}^3$, it can be expressed as the torus with inner circle $x^2+y^2=1$ and the outer circle $x^2+y^2=9$, and the part with $x\ge 2$ is chopped off.

\para{Sampling parameters.}
As stated in the main Theorem~\ref{main}, the lower bound of sampling that guarantees deformation retraction with probability $1-\gamma$ can be expressed as $$n^* = \beta(\epsilon)(\ln \beta(\frac{\epsilon}{2})+\ln(\frac{1}{\gamma})),$$
where $\displaystyle\beta(x):=\frac{\vol(\m)}{\frac{\cos^k\theta}{2^{k+1}}I_{1-\frac{x^2\cos^2\theta}{16\delta^2}}(\frac{k+1}{2}, \frac{1}{2})\vol(B_{x}^k)}$ and $\theta=\arcsin(\frac{x}{4\delta})$.

For \emph{cylinder}, $\vol(\m) = 2\pi$, $k=2$, $\vol(B_{x}^2) = \pi x^2$, $\delta = 1$. 
For instance, setting $\epsilon=0.49$, $\gamma = 0.1$ gives rise to $n^* = 638$, as illustrated in Fig.~\ref{fig:data} (a). 

For \emph{torus}, $\vol(\m) = (8-0.522) \cdot 2 \pi$, $k=2$, $\vol(B_{x}^2) = \pi x^2$, $\delta = 1$. 
For instance, setting $\epsilon=0.49$, $\gamma = 0.1$ gives rise to $n^* = 9809$, as illustrated in Fig.~\ref{fig:data} (b). 

\begin{figure}[!ht]
\begin{center}
\begin{tabular}{cc}
\includegraphics[width=0.45\linewidth]{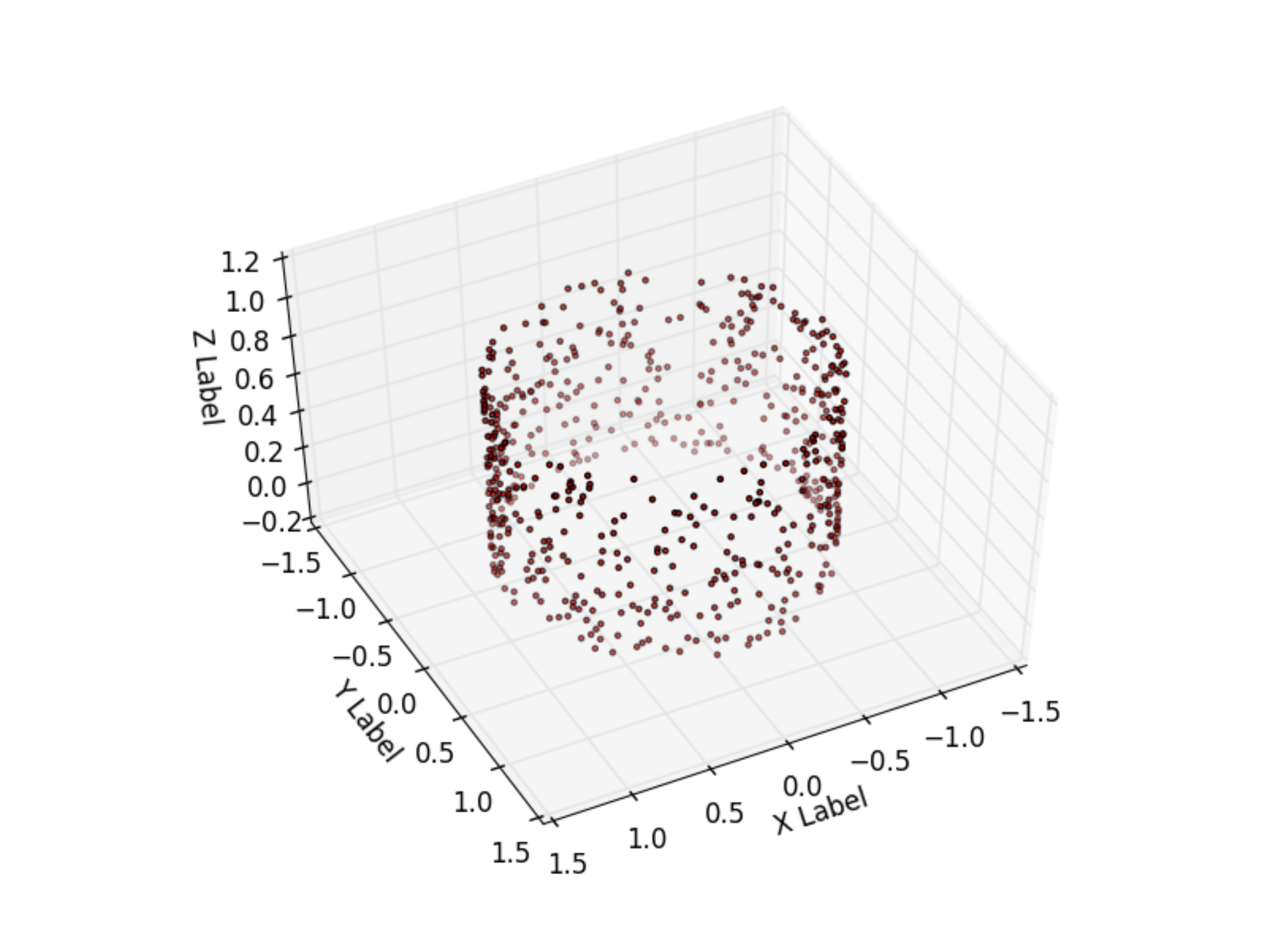} &
\includegraphics[width=0.45\linewidth]{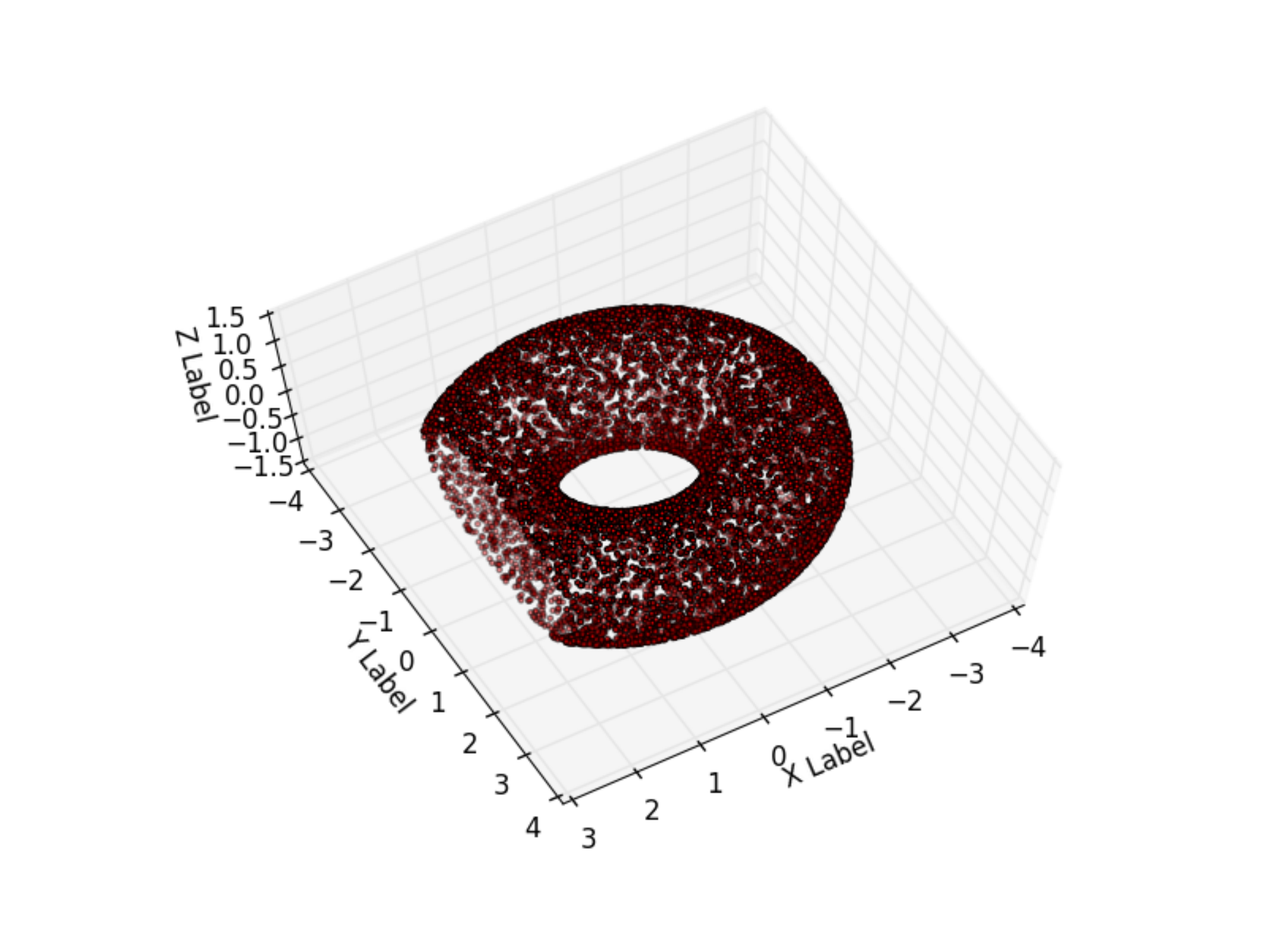}\vspace{-2mm}\\
{\bf (a)} & {\bf (b)}\\
\end{tabular}
\vspace{-2mm}
\caption{Point cloud samples for \emph{cylinder} {\bf (a)} and \emph{torus} {\bf (b)}.}
\label{fig:data}
\end{center}
\end{figure}

\para{Distribution of lower bounds.}
For a fixed sample quality $\epsilon$, we demonstrate the distribution of lower bounds $n^*$ as $\gamma$ increases from $0.05$ to $0.95$ (that is, confidence ranges from $95\%$ to $5\%$). This is shown in Fig.~\ref{fig:nstar}. Intuitively, for a fixed sample quality, we need more point samples in order to obtain higher confidence in topological inference. 

\begin{figure}[!ht]
\begin{center}
\begin{tabular}{cc}
\includegraphics[width=0.4\linewidth]{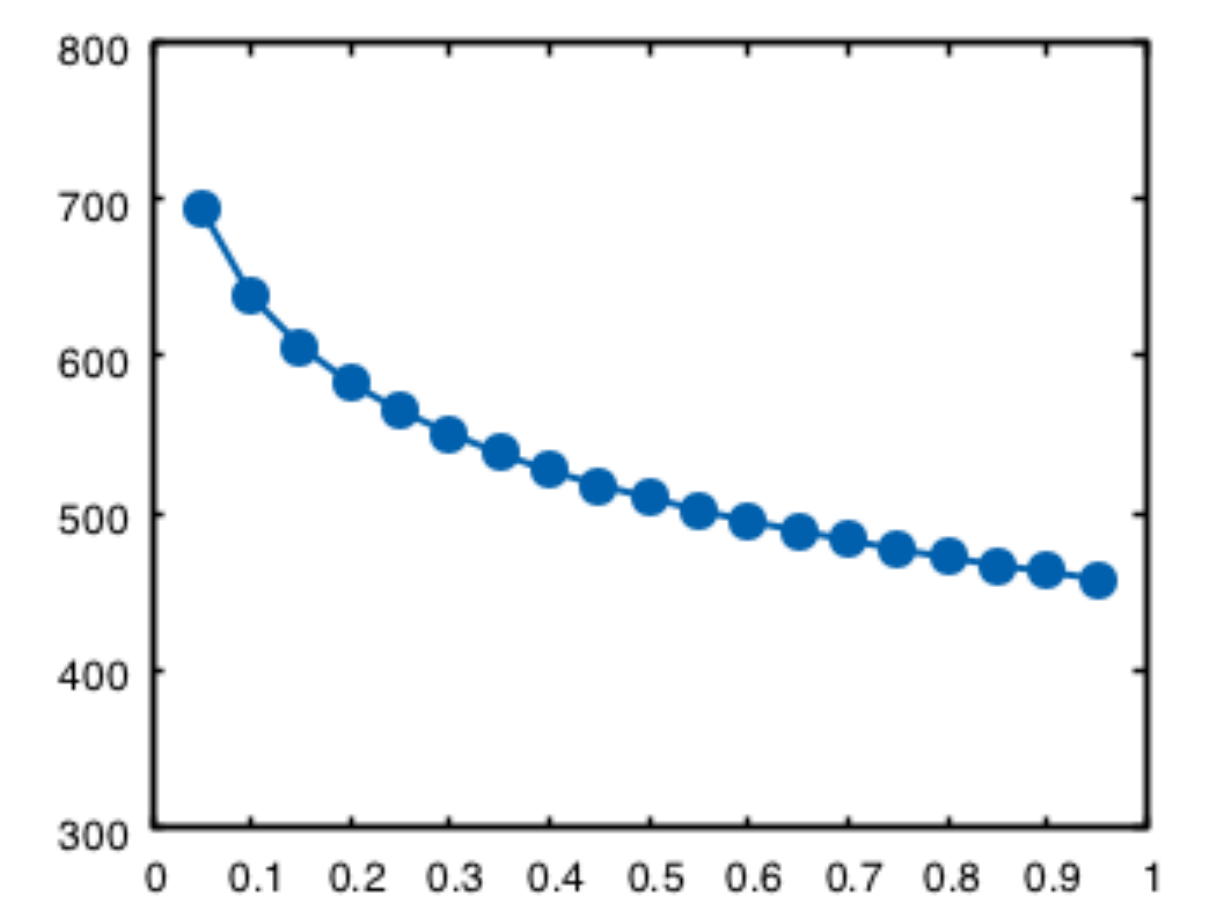} &
\includegraphics[width=0.4\linewidth]{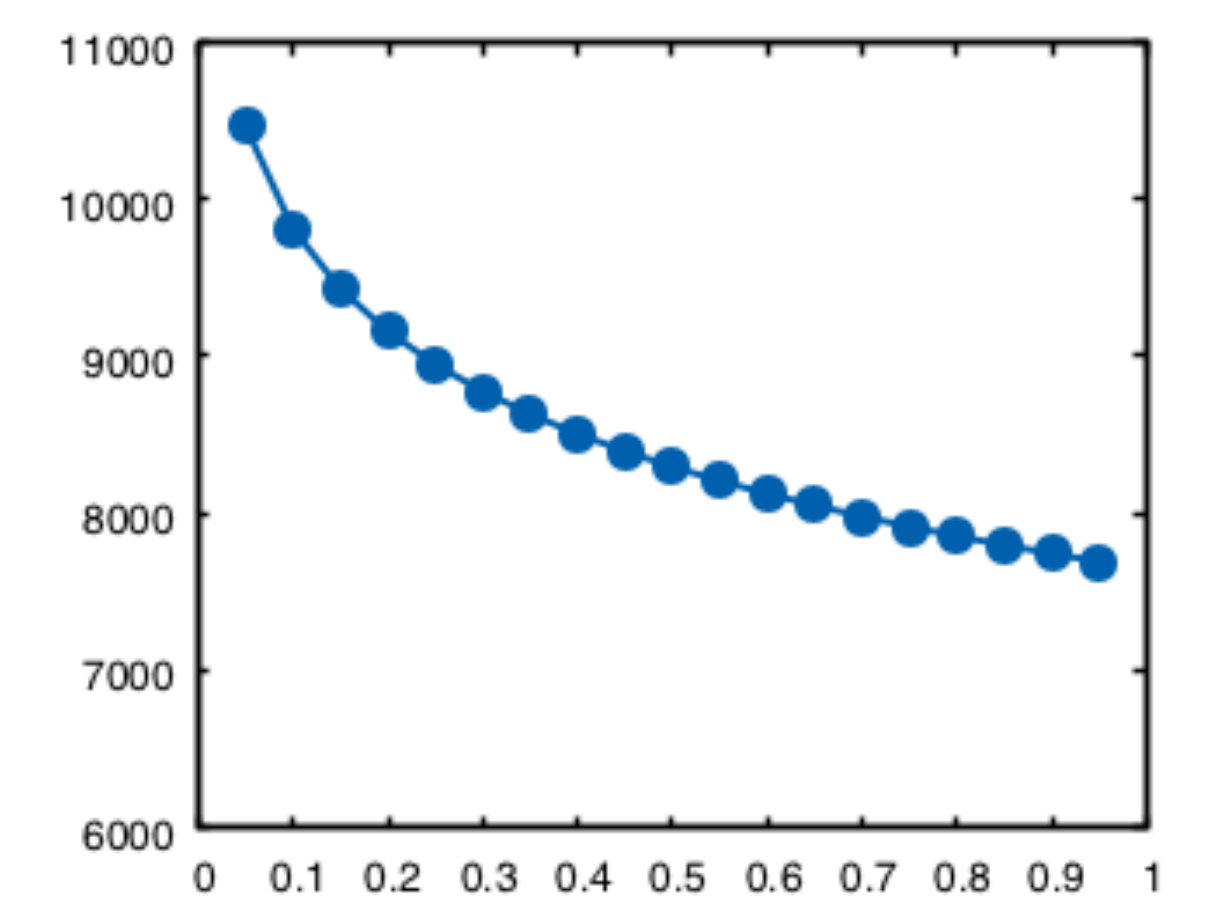}\vspace{-2mm}\\
{\bf (a)} & {\bf (b)}\\
\end{tabular}
\vspace{-2mm}
\caption{Lower bounds for \emph{cylinder} {\bf (a)} and \emph{torus} {\bf (b)} for a fixed $\epsilon =0.49$,  x-axis corresponds to $\gamma$ while y-axis corresponds to $n^*$.}
\label{fig:nstar}
\end{center}
\end{figure}

Meanwhile, for a fixed $\gamma = 0.1$, which corresponds to a confidence of $90\%$, we illustrate the distribution of lower bounds $n^*$ as $\epsilon$ increases from $0.15$ to $0.5$. 
By Theorem~\ref{main}, it is rather obvious that we need more points to have higher quality samples for a fixed confidence level. This is shown in Fig.~\ref{fig:samplequality}.

\begin{figure}[!ht]
\begin{center}
\begin{tabular}{cc}
\includegraphics[width=0.4\linewidth]{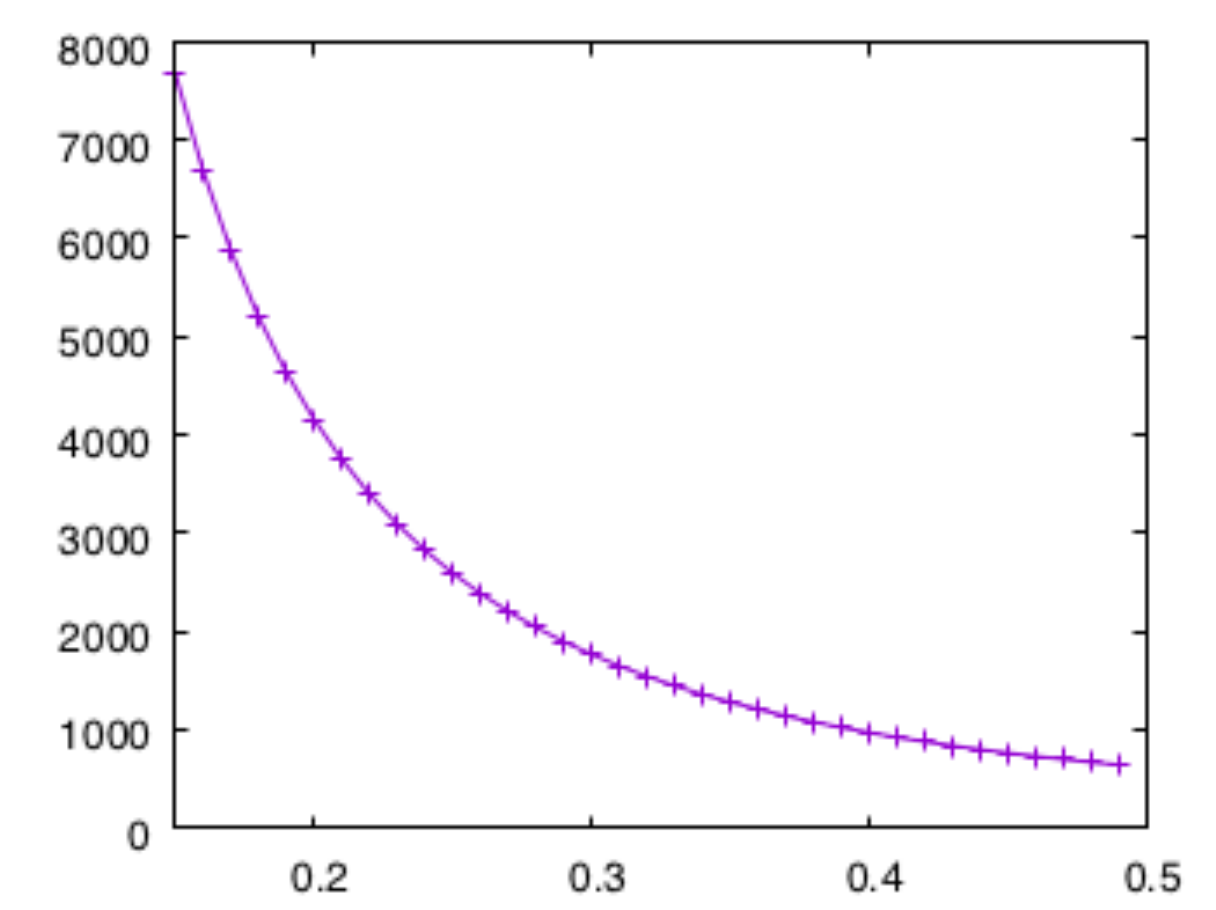} &
\includegraphics[width=0.4\linewidth]{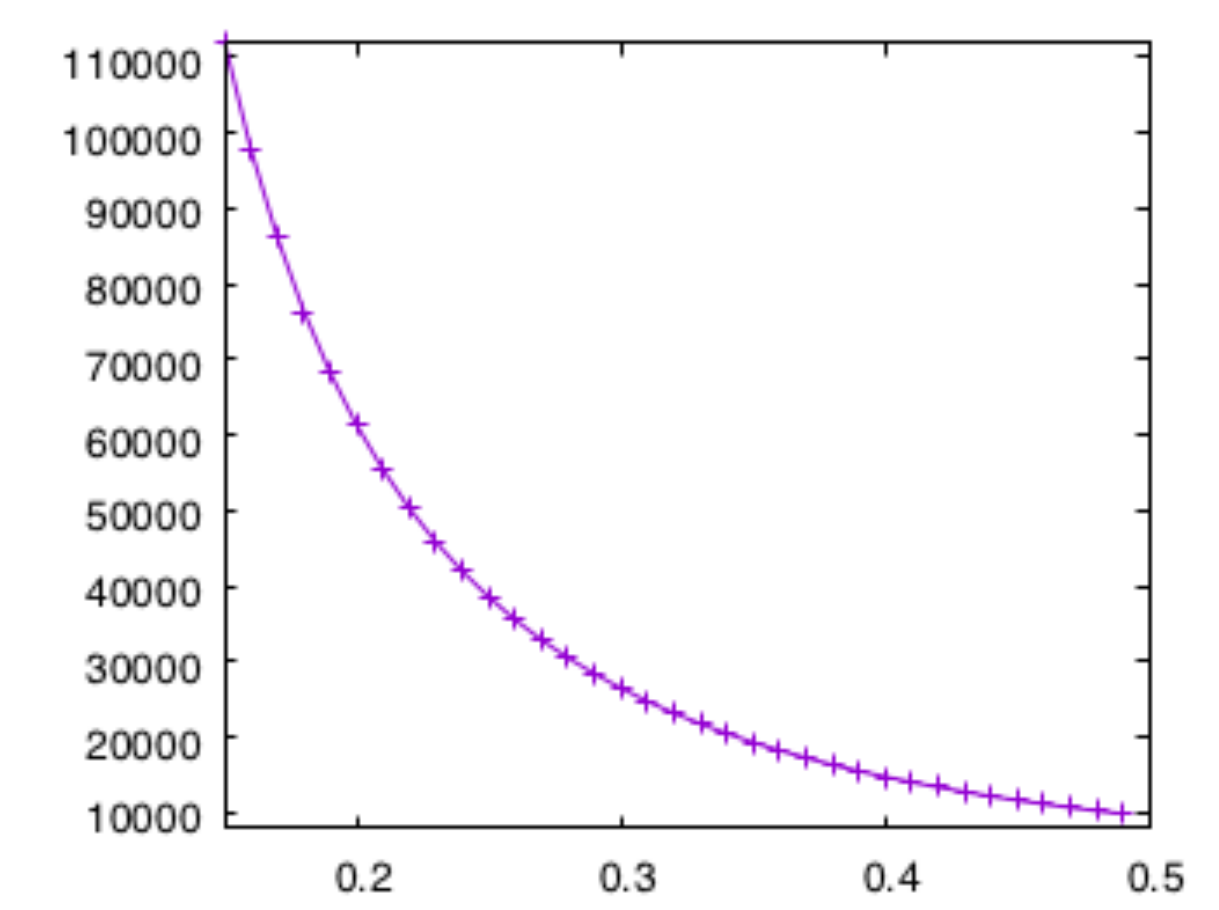}\vspace{-2mm}\\
{\bf (a)} & {\bf (b)}\\
\end{tabular}
\vspace{-2mm}
\caption{Lower bounds for \emph{cylinder} {\bf (a)} and \emph{torus} {\bf (b)} for a fixed $\gamma =0.1$,  x-axis corresponds to the $\epsilon$ while y-axis corresponds to $n^*$.}
\vspace{-2mm}
\label{fig:samplequality}
\end{center}
\end{figure}

\begin{figure}[!ht]
\begin{center}
\begin{tabular}{cc}
\includegraphics[width=0.47\linewidth]{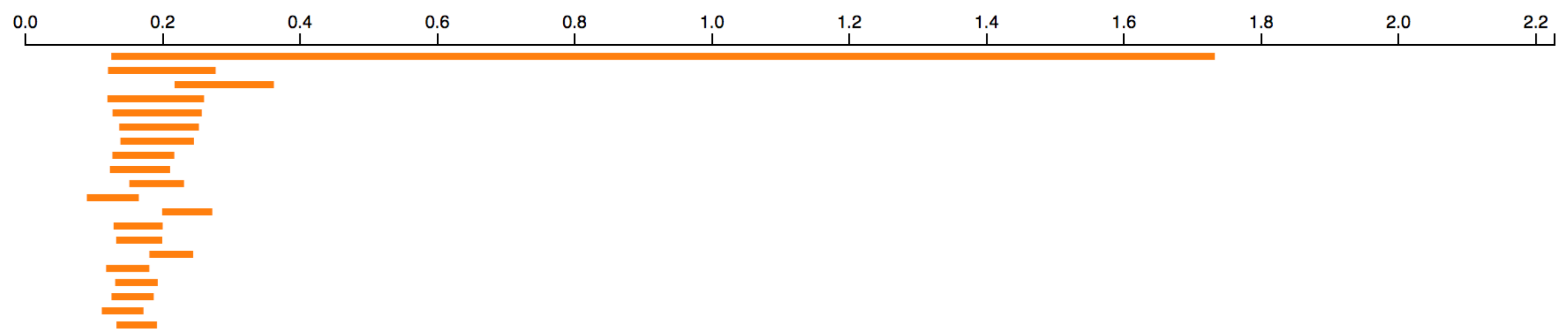} &
\includegraphics[width=0.47\linewidth]{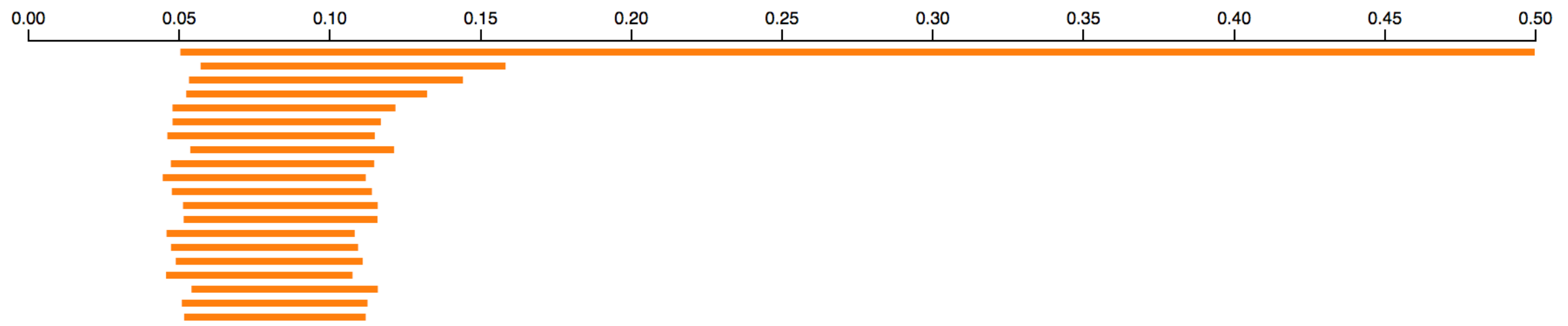} \vspace{-2mm}\\
{\bf (a)} & {\bf (d)}\\
\includegraphics[width=0.47\linewidth]{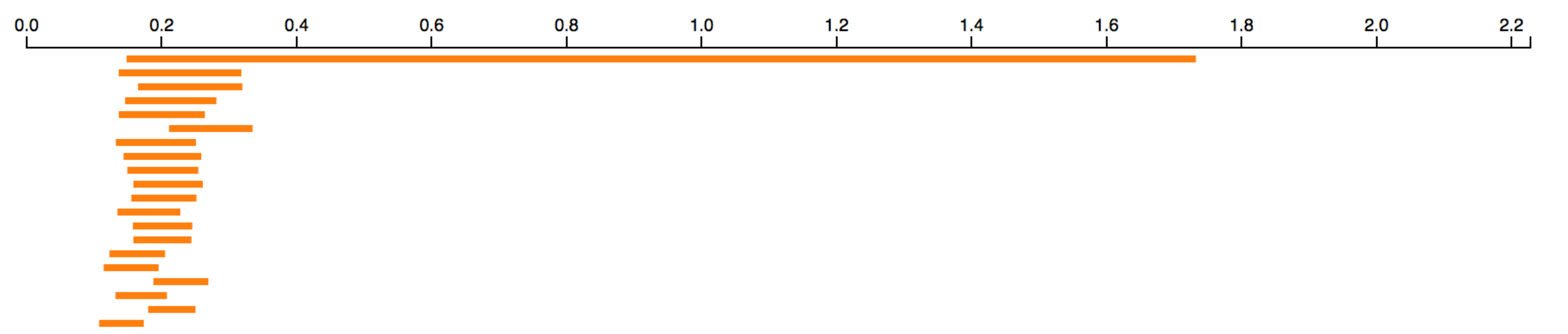} &
\includegraphics[width=0.47\linewidth]{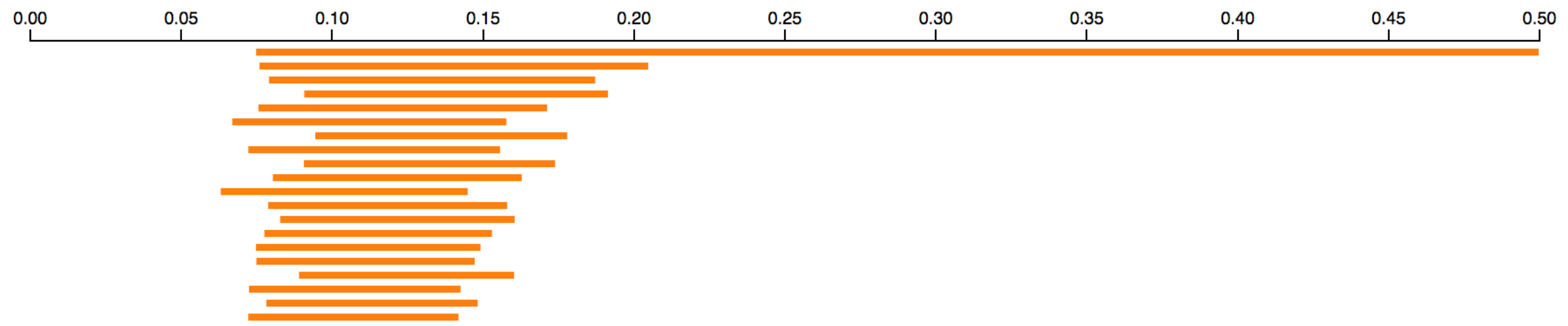} \vspace{-2mm}\\
{\bf (b)} & {\bf (e)}\\
\includegraphics[width=0.47\linewidth]{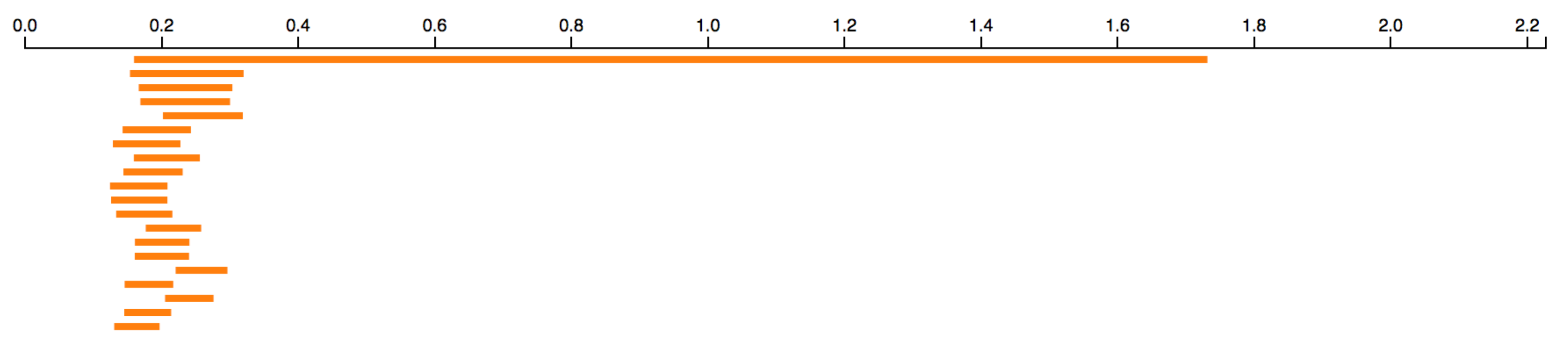} &
\includegraphics[width=0.47\linewidth]{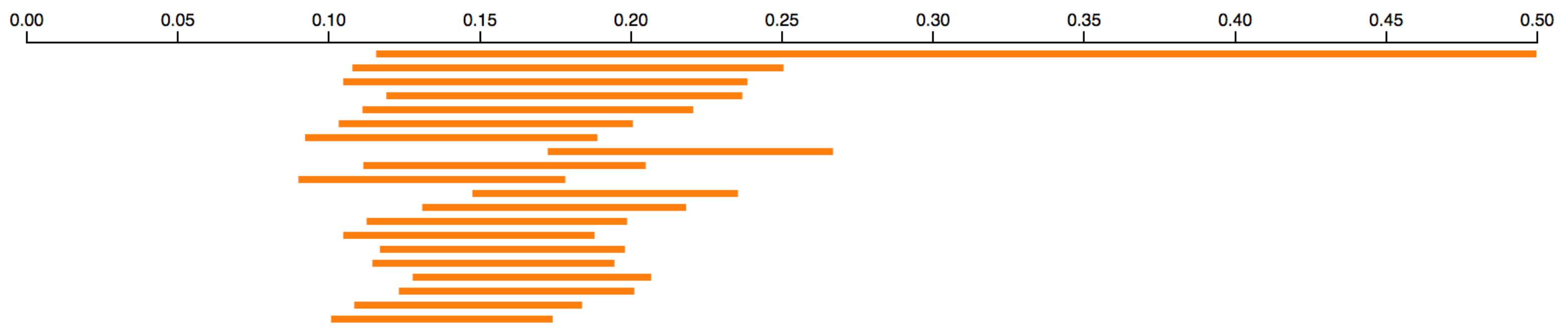} \vspace{-2mm}\\
{\bf (c)} & {\bf (f)}\\
\end{tabular}
\vspace{-2mm}
\caption{Persistent barcodes for \emph{cylinder}. Each plot only shows the top 20 most persistent (longest) cycles. 
For a fixed $\epsilon = 0.49$, $r \in [0, 2.2]$: {\bf (a)} $\gamma = 0.1$, $n^* = 638$; {\bf (b)} $\gamma = 0.1$, $n^* = 583$; and {\bf (c)} $\gamma = 0.3$, $n^* = 551$. 
For a fixed $\gamma=0.1$ ($90\%$ confidence), $r \in [0, 0.5]$: {\bf (d)} $\epsilon = 0.2$, $n^* = 4160 $; {\bf (e)} $\epsilon = 0.3$, $n^* = 1763$; and {\bf (f)} $\epsilon = 0.4$, $n^* = 967$. Notice different scales between the plots on the left and on the right. 
}
\vspace{-2mm}
\label{fig:cylinder-ph}
\end{center}
\end{figure}

\para{Homology computation.} 
Finally we can perform homology computation on the above point clouds; in particular, for a given sample $\overline{x}$ and its corresponding $U$, we show that the homology of $U$ equals the homology of $\m$. Admittedly,  homology is a very weak verification of our main sampling theorem. In fact, if one's goal is only to recover the same homology of a manifold with point  samples, our estimation from Theorem~\ref{main} is an obvious overestimation.  
In other words, our estimation of the lower bound $n^*$ has to account for the boundary condition and to guarantee deformation retract (not just homology or homotopy equivalence). 

Nevertheless, we show the results of homology inference across multiple $\gamma$ with a fixed $\epsilon$, as well as the results across multiple $\epsilon$ with a fixed $\gamma$. We rely on the computation of persistent homology to recover the homological information of a point cloud sample. 
Persistent homology, roughly speaking, operates on a point cloud sample $\overline{x}$ and tracks how the homology of $U(r)=\bigcup_{x\in\overline{x}}B_{r}(x)$ changes as $r$ increases (where typically $r \in [r_0=0, r_k]$, for some positive real value $r_k$). 
Specifically, it applies the homology functor $\Hgroup$ to a sequence of topological spaces connected by inclusions, 
$$U(r_0) \to \cdots \to U(r_i) \to U(r_{i+1}) \to \cdots \to U(r_k),$$ 
and studies a multi-scale notion of homology, 
$$\Hgroup(U(r_0)) \to \cdots \to \Hgroup(U(r_i)) \to \Hgroup(U(r_{i+1})) \to \cdots \to \Hgroup(U(r_k)),$$
see~\cite{EdelsbrunnerHarer2008, EdelsbrunnerHarer2010, Ghrist2008} for introduction to persistent homology. 
We use the software package Ripser~\cite{Bauer2016} for the computation of persistent homology. 
Given a point cloud sample $\overline{x}$, Ripser computes its persistent homology using Vietoris--Rips complexes formed 
on $\overline{x}$ and encodes the homological information using persistence barcodes. 
In a nutshell, each bar in the persistence barcodes captures the time when a homology class appears and disappears as $r$ increases. 
As the homology of a union of balls is guaranteed (by the Nerve Lemma) to be the one of a \v{C}ech complex, 
the results of~\cite{AttaliLieutierSalinas2013} could be utilized to deduce results on a Vietoris--Rips complex from a \v{C}ech complex. 

For \emph{cylinder} dataset, the $1$-dimensional homology of its underlying manifold should be of rank one, as the dataset contains one significant cycle (tunnel). 
For a fixed $\epsilon = 0.49$, we compute the $1$-dimensional persistent homology of the point clouds at parameter $\gamma = 0.1, 0.2, 0.3$ respectively. Their persistent barcodes are shown in Fig.~\ref{fig:cylinder-ph}(a)-(c) respectively. For each plot, the longest bar corresponds to the most significant $1$-dimensional cycle, which clearly corresponds to the true homological feature of the underlying manifold.

Meanwhile, the $1$-dimensional homology of the manifold underlying the \emph{torus} dataset (with a cap chopped off) should be of rank two, as the dataset contains two significant cycles (same as the classic torus dataset).
We have similar results as in the case of \emph{cylinder} dataset. For simplicity, we give the persistent barcodes for $\epsilon = 0.49$, $\gamma = 0.2$, $n^* = 9157$ in Fig.~\ref{fig:torus-ph}. 
Here, the first two longest bars correspond to the two most significant $1$-dimensional cycles, which again clearly correspond to the true homological features of the underlying manifold.  

\begin{figure}[!ht]
\begin{center}
\begin{tabular}{c}
\includegraphics[width=0.8\linewidth]{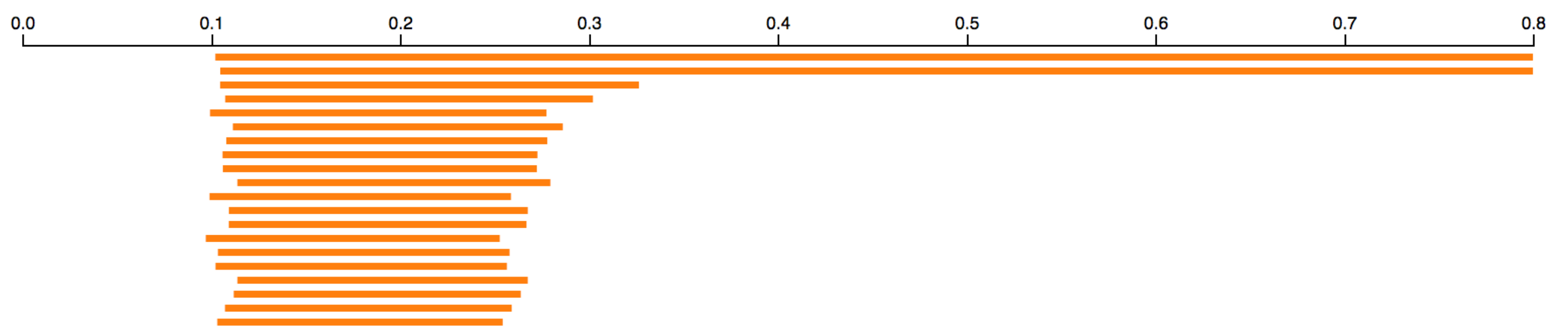} 
\end{tabular}
\vspace{-2mm}
\caption{Persistent barcodes for \emph{torus}, for $\epsilon = 0.49$, 
$\gamma = 0.2$, $n^* = 9157$, $r \in [0,0.8]$. Only the top 20 most persistent (longest) cyclesare shown.}
\vspace{-2mm}
\label{fig:torus-ph}
\end{center}
\end{figure}

\section{Discussions}
\label{sec:discussions}

Given a point cloud sample of a compact, differentiable manifold with boundary, we give a probabilistic notion of sampling condition that is not handled by existing theories. 
Our main results relate topological equivalence between the offset and the manifold as a deformation retract, which is stronger than homological or homotopy equivalence. Many interesting questions remain. 

First, while our sampling condition considers differentiable manifolds with boundary, it cannot be trivially extended to handle manifolds with corners. 
The fundamental difficulty arises because the $\reach$ becomes zero in the case of manifolds with corners. We suspect that deriving practical sampling conditions for manifolds with corners, and in general, for stratified spaces, is challenging and requires new way of thinking. 

Second, we have conducted experiments that verify homological equivalence between the offset of samples and the underlying manifold. However, such an experiment is a very weak verification of our main inference theorem. 
Experimentally computing or verifying deformation retract in the point cloud setting (as stated in Theorem~\ref{main}), possibly via the study of discrete gradient fields, remains an open question. 

\bibliographystyle{abbrv}
\bibliography{manifold_boundary}  
\end{document}